\documentclass[11pt,a4paper,reqno,intlimits,sumlimits]{amsart}

\usepackage{amsmath}
\usepackage{amssymb}
\usepackage{chicago}
\usepackage[mathscr]{euscript}
\usepackage{enumerate}
\usepackage{xspace}
\usepackage{color}
\usepackage{epsfig,rotating}

\DeclareMathAlphabet{\mathpzc}{OT1}{pzc}{m}{it}

\begin{document}
\theoremstyle{plain}
\newtheorem{theorem}{Theorem}[section]
\newtheorem{lemma}[theorem]{Lemma}
\newtheorem{proposition}[theorem]{Proposition}
\newtheorem{corollary}[theorem]{Corollary}

\theoremstyle{definition}
\newtheorem{remark}[theorem]{Remark}
\newtheorem{example}[theorem]{Example}
\newtheorem{assumption}[theorem]{Assumption}
\newtheorem{statement}[theorem]{Statement}
\newtheorem*{assuEM}{Assumption ($\mathbb{EM}$)}
\newtheorem*{IREM}{Important Remark}

\newcommand{\Law}{\ensuremath{\mathop{\mathrm{Law}}}}
\newcommand{\loc}{{\mathrm{loc}}}
\newcommand{\Log}{\ensuremath{\mathop{\mathcal{L}\mathrm{og}}}}
\newcommand{\Meixner}{\ensuremath{\mathop{\mathrm{Meixner}}}}

\let\SETMINUS\setminus
\renewcommand{\setminus}{\backslash}

\def\stackrelboth#1#2#3{\mathrel{\mathop{#2}\limits^{#1}_{#3}}}

\renewcommand{\theequation}{\thesection.\arabic{equation}}
\numberwithin{equation}{section}

\newcommand{\prozess}[1][L]{{\ensuremath{#1=(#1_t)_{0\le t\le T}}}\xspace}
\newcommand{\prazess}[1][L]{{\ensuremath{#1=(#1_t)_{0\le t\le T^*}}}\xspace}
\newcommand{\scal}[2]{\ensuremath{\langle #1, #2 \rangle}}
\def\C{\ensuremath{\mathbb{C}}}
\def\F{\ensuremath{\mathcal{F}}}
\def\R{\ensuremath{\mathbb{R}}}
\def\bF{\mathbf{F}}
\def\V{\mathbb{V}}

\def\supL{\overline{L}}
\def\infL{\underline{L}}

\def\Rdmz{\R^d\setminus\{0\}}
\def\Rmz{\R\setminus\{0\}}
\def\Rnmz{\R^n\setminus\{0\}}
\def\Rp{\mathbb{R}_{\geqslant0}}

\def\la{\ensuremath{L^1(\R)}}
\def\labc{\ensuremath{L^1_{\text{bc}}(\R)}}
\def\lad{\ensuremath{L^1(\R^d)}}
\def\lap{\ensuremath{L^\infty(\R)}}
\def\lat{\ensuremath{L^2(\R)}}

\def\lev{L\'{e}vy\xspace}
\def\lk{L\'{e}vy--Khintchine\xspace}
\def\mg{martingale\xspace}
\def\num{num\'{e}raire\xspace}
\def\smmg{semimartingale\xspace}
\def\tih{time-inhomogeneous\xspace}
\def\wh{Wiener--Hopf\xspace}

\def\eqlaw{\ensuremath{\stackrel{\mathrrefersm{d}}{=}}}
\def\chartri{\ensuremath{(b,\sigma,\nu)}}
\def\dsdx{\ensuremath{(\ud s, \ud x)}}
\def\dtdx{\ensuremath{(\ud t, \ud x)}}

\def\intrr{\ensuremath{\int_{\R}}}

\def\ud{\ensuremath{\mathrm{d}}}
\def\dt{\ud t}
\def\ds{\ud s}
\def\dx{\ud x}
\def\dy{\ud y}
\def\dz{\ud z}
\def\du{\ud u}
\def\e{\mathrm{e}}
\def\ecc{\mathbf{e}_\mathpzc{i}}
\def\icc{\mathpzc{i}}

\def\EM{\ensuremath{(\mathbb{EM})}\xspace}
\def\ES{\ensuremath{(\mathbb{ES})}\xspace}
\def\AC{\ensuremath{(\mathbb{AC})}\xspace}

\def\ott{{0\leq t\leq T}}

\def\bg{\ensuremath{\bar g}}
\def\logs{\mathpzc s}

\title[Wiener--Hopf factorization and exotic options]
      {Analyticity of the Wiener--Hopf factors and valuation of exotic options in L\'evy models}

\author[E. Eberlein]{Ernst Eberlein}
\author[K. Glau]{Kathrin Glau}
\author[A. Papapantoleon]{Antonis Papapantoleon}

\address{Department of Mathematical Stochastics, University of Freiburg,
        Eckerstr. 1, 79104 Freiburg, Germany}
\email{eberlein@stochastik.uni-freiburg.de}

\address{Department of Mathematical Stochastics, University of Freiburg,
        Eckerstr. 1, 79104 Freiburg, Germany}
\email{glau@stochastik.uni-freiburg.de}

\address{Institute of Mathematics, TU Berlin, Stra\ss e des 17. Juni 136,
         10623 Berlin, Germany \& Quantitative Products Laboratory,
         Deutsche Bank AG, Alexanderstr. 5, 10178 Berlin, Germany}
\email{papapan@math.tu-berlin.de}

\thanks{K.~G. would like to thank the DFG for financial support through project
        \mbox{EB66/11-1}, and the Austrian Science Fund (FWF)
        for an invitation under grant P18022. A.~P. gratefully acknowledges the
        financial support from the Austrian Science Fund (FWF grant Y328,
        START Prize)}
\date{}\maketitle
\setcounter{tocdepth}{1}
\frenchspacing
\pagestyle{myheadings}

\begin{abstract}
This paper considers the valuation of exotic path-dependent options in \lev
models, in particular options on the supremum and the infimum of the asset
price process. Using the \wh factorization, we derive expressions for the
analytically extended characteristic function of the supremum and the infimum
of a \lev process. Combined with general results on Fourier methods for option
pricing, we provide formulas for the valuation of one-touch options, lookback
options and equity default swaps in \lev models.
\end{abstract}

\section{Introduction}

The ever-increasing sophistication of derivative products offered by financial
institutions, together with the failure of traditional Gaussian models to describe
the dynamics in the markets, has lead to a quest for more realistic and flexible
models. In fact one of the lessons from the current financial crisis is the
following: the Gaussian copula model is inappropriate to describe the
interdependence between the tails of asset returns because, among other pitfalls,
the tail dependence coefficient is always zero; hence, this model cannot capture
systemic risk.

In the search for appropriate alternatives, \emph{\lev processes}
are playing a leading role, either as models for financial assets themselves,
or as building blocks for models, e.g. in \lev-driven stochastic volatility
models or in affine models. The field of \lev processes has become popular
 in modern mathematical finance, and the interest from academics
and practitioners has led to inspiring and challenging questions.

\lev processes are attractive for applications in mathematical finance because
they can describe some of the observed phenomena in the markets in a rather
adequate way. This is due to the fact that their sample paths may have jumps and
the generated distributions can be heavy-tailed and skewed. Another important
improvement concerns the famous smile effect. See Eberlein and Keller
\citeyear{EberleinKeller95} and
\citeN{EberleinPrause02} for an extensive empirical justification of the
non-Gaussianity of asset returns and the appropriateness of (generalized
hyperbolic) \lev processes.  For an overview of the application of \lev
processes in finance the interested reader is referred to the text\-books of
\citeN{ContTankov03}, \citeN{Schoutens03} as well as the collection edited by
\shortciteN{KyprianouSchoutensWilmott05}. There are, of course, several textbooks
dealing with the theory of \lev processes; we mention \citeN{Bertoin96},
Sato \citeyear{Sato99}, \citeN{Applebaum04} and \citeN{Kyprianou06}, while the collection
by \shortciteN{Barndorff-NielsenMikoschResnick01} contains an overview of the
application of \lev processes in different areas of research, such as quantum
field theory and turbulence.

The application of \lev processes in financial modeling, in particular for
the pricing and hedging of derivatives, has led to new challenges of both
\emph{analytical} and \emph{numerical} nature. In \lev models simple closed
form valuation formulas are typically not available even for plain vanilla
European options, let alone for exotic path-dependent options. The numerical
methods which have been developed in the classical Gaussian framework lead
to completely new challenges in the context of \lev driven models. These
numerical methods can be
classified roughly in three areas: probabilistic numerical methods (Monte
Carlo methods), deterministic numerical methods (PIDE methods), and Fourier
transform methods; for an excellent survey of these methods, their
applicability and limitations, we refer to \shortciteN{HilberReichSchwabWinter09}.

This paper focuses on the application of Fourier transform methods for the
valuation of exotic path-dependent options, in particular options depending
on the supremum and the infimum of \lev processes. The bulk of the literature
on this latter topic focuses on the numerical aspects. Our focus is on the
analytical aspects. More specifically, we show first that the \wh factorization
of a \lev process possesses an analytic extension, and then we prove that the \wh
factorization (viewed as a Laplace transform in time) can be inverted. These
results allow us to derive expressions for the extended characteristic function
of the supremum and the infimum of a \lev process. This latter result, combined
with general results on option pricing by Fourier methods (cf.
\shortciteNP{EberleinGlauPapapantoleon08}), allows us to derive pricing formulas
for lookback options, one-touch options and equity default swaps in \lev models.

Let us briefly comment on some papers where the \wh factorization is used to
price exotic options in \lev models. Boyarchenko and Levendorski\v{\i}
\citeyear{BoyarchenkoLevendorskii02b} derive
valuation formulas for barrier and one-touch options for driving \lev processes
that belong to the class of so-called ``regular \lev processes of exponential type''
(RLPE); cf. also the book by Boyarchenko and Levendorski\v{\i}
\citeyear{BoyarchenkoLevendorskii02book}. The results of
these authors are based on the theory of pseudodifferential operators. The
numerics of this approach is pushed further in \citeANP{KudryavtsevLevendorskii09}
\citeyear{KudryavtsevLevendorskii06,KudryavtsevLevendorskii09}.
\shortciteN{AvramKyprianouPistorius04}, \shortciteN{AsmussenAvramPistorius04},
\citeN{KyprianouPistorius03}, Alili and Kyprianou \citeyear{AliliKyprianou05}, and
\shortciteN{LevendorskiiKudryavtsevZherder05} consider the
valuation of American and Russian options, either on a finite or an infinite
time horizon. Jeannin and Pistorius \citeyear{JeanninPistorius09} develop methods for
the computation of prices and Greeks for various \lev models. Central in
their argumentation is the approximation of different \lev models by the
class of  ``generalized hyper-exponential L\'evy models'', which
have a tractable \wh factorization. The same approach is also applied
in \shortciteN{AsmussenMadanPistorius05} for the pricing of equity default
swaps in \lev models.

The major open challenge in this field is the development of analytical expressions
for the \wh factors for general \lev processes. In a remarkable recent development,
\citeN{HubalekKyprianou10} generate a family of spectrally negative \lev processes
with tractable \wh factors, using results from potential theory for subordinators.
These results were later extended in \citeN{KyprianouRivero08} and applied to problems
in actuarial mathematics in \shortciteN{KyprianouRiveroSong09}. Moreover, in two
very recent papers \citeANP{Kuznetsov09a} \citeyear{Kuznetsov09a,Kuznetsov09b}
introduces special families of \lev processes such that the Wiener--Hopf factors
can be computed as infinite products over the roots of certain transcendental
equations. These families include processes with behavior similar to the CGMY
process, while the author shows that the numerical computation of the infinite
products can be performed quite efficiently.

This paper is structured as follows: in section \ref{ch3:sup-Levy}, we briefly
review \lev processes and prove the analyticity of the characteristic function
of the supremum. In section \ref{wh-section}, we review the \wh factorization,
prove its analytic extension and invert it in time. In section \ref{Levy-ex-prop},
we present some examples of popular \lev models and comment on the continuity
of their laws. Finally, in section \ref{ch3:sc5}, we derive valuation formulas
for lookback and one-touch options as well as for equity default swaps.

\begin{IREM}
This paper is intimately tied to, and intended to be read together with, the
companion paper \citeANP{EberleinGlauPapapantoleon08} 
\citeyear{EberleinGlauPapapantoleon08}, which will be abbreviated
EGP in the sequel. In particular, we will make heavy use of the notation and
results from that paper.
\end{IREM}

\section{L\'evy processes}
\label{ch3:sup-Levy}

We start by fixing the notation that will be used throughout the paper and providing
some estimates on the exponential moments of a \lev process. Then, we prove the
analytic extension of the characteristic function of the supremum and the infimum
of a \lev process, sampled either at a fixed time or at an independent, exponentially
distributed time.

\subsection{Notation}

Let $\mathscr{B}=(\Omega,\F,\bF,P)$ be a complete stochastic basis in the sense
of \citeN[I.1.3]{JacodShiryaev03}, where $\F=\F_T$, $0<T\le\infty$ and
$\bF=(\F_t)_{0\le t\le T}$.
Let $L=(L_t)_{0\le t\le T}$ be a \emph{\lev process} on this stochastic basis,
i.e. $L$ is a semimartingale with \emph{independent} and \emph{stationary
increments} (PIIS), and $L_0=0$ a.s. 
We denote the \emph{triplet of predictable characteristics}
of $L$ by $(B,C,\nu)$ and the \emph{triplet of local characteristics} by
($b,c,\lambda$); using \citeN[II.4.20]{JacodShiryaev03} the two triplets are
related via
\begin{displaymath}
B_t(\omega)=bt, \quad
C_t(\omega)=ct, \quad
\nu(\omega;\ud t,\ud x)=\lambda(\ud x)\,\dt.
\end{displaymath}
We assume that the following condition is in force.
\begin{assuEM}
There exists a constant $M>1$ such that
\begin{align*}
\int_{\{|x|>1\}} \e^{ux}\lambda(\dx)<\infty,
  \qquad\quad
 \forall u\in[-M,M].
\end{align*}
\end{assuEM}
The triplet of predictable characteristics of a PIIS determines the law of
the random variables; more specifically, for a \lev process we know from
the \lk formula that
\begin{align}\label{lk-form}
E\left[\e^{iuL_t}\right]
= \exp\big(t\cdot\kappa(iu)\big),
\end{align}
for all $t\in[0,T]$ and all $u\in\R$, where the cumulant generating function
is
\begin{align}\label{levy-cmu}
\kappa(u) = ub + \frac{u^2}2c
          + \int_\R(\e^{ux}-1-ux)\lambda(\dx).
\end{align}

Assumption ($\mathbb{EM}$) entails that the \lev process $L$ is a \emph{special}
and \emph{exponentially special} semimartingale, hence the use of a truncation
function can be and has been omitted. Applying Theorem 25.3 in \citeN{Sato99} we
get that
\begin{align*}
E \left[ \e^{uL_t} \right] <\infty,
  \qquad\quad
 \forall u\in[-M,M], \;\; \forall t\in[0,T].
\end{align*}

Recall that for any stochastic process $X$ we denote by $\overline{X}$ the
supremum and by $\underline{X}$ the infimum process of $X$ respectively.

In the sequel, we will provide the proofs of the results for the supremum
process. The proofs for the infimum process can be derived analogously or using
the duality between the supremum and the infimum process; see the following
remark.

\begin{remark}\label{dual-Levy}
Let $L$ be a \lev process with local characteristics $(b,c,\lambda)$.
The \emph{dual} of the \lev process $L$ defined by $L':=-L$, has the
triplet of local characteristics $(b',c',\lambda')$ where $b'=-b$,
$c'=c$ and $1_A(x)*\lambda'=1_A(-x)*\lambda$, $A\in\mathcal B(\Rmz)$.
Moreover, we have that
\begin{align*}
\infL_t = \inf_{0\le s\le t} L_s
        = - \sup_{0\le s\le t} (-L_s)
        = -\overline{L'}_t \,.
\end{align*}
\end{remark}

\subsection{Analytic extension, fixed time case}
In this section, we establish the existence of an analytic extension of the
characteristic function of the \emph{supremum} and the \emph{infimum} of a
\lev process, and derive explicit bounds for the exponential moments of the
supremum and infimum process.

The next lemma endows us with a link between the existence of exponential
moments of a measure $\varrho$ and the analytic extension of the characteristic
function $\widehat{\varrho}$.

\begin{lemma}\label{lem:sec5:neu1}
Let $\varrho$ be a measure on the space $(\R,\mathcal{B}(\R))$. If
$\int\e^{ux}\varrho(\dx)< \infty$ for all
$u\in [-a,b]$ with $a,b\ge0$, then the characteristic function $\widehat\varrho$
has an extension that is continuous on $(-\infty,\infty)\times i[-b,a]$
and is analytic in the interior of the strip, $(-\infty,\infty)\times i(-b,a)$.
Moreover $\widehat\varrho(u)=\int\e^{iux}\varrho(\dx)$ for all $u \in \C$
with $\Im(u)\in[-b,a]$.
\end{lemma}
\begin{proof}
The function $u\mapsto\e^{iux}$ clearly extends to an entire function and
the extension
\begin{align*}
\widehat\varrho(u):= \int\e^{iux}\varrho(\dx)
 \quad \big(u\in\C \text{ with }\Im(u)\in [-b,a]\big)
\end{align*}
is well-defined since
\begin{align*}
\left|\e^{iux}\right|
 = \e^{-\Im(u)x}\le \e^{-ax}1_{\{x\le 0\}} + \e^{bx}1_{\{x>0\}}
 =: h(x),
\end{align*}
for $u \in \C$ with $\Im(u)\in [-b,a]$, and we have that $h\in L^1(\varrho)$
by assumption. Moreover, Lebesgue's dominated convergence theorem yields that
this extension is continuous.

We will prove the analyticity of $\widehat\varrho$ in
$(-\infty,\infty)\times i(-b,a)$ using the theorem of Morera (cf. for example
Theorem 10.17 in \citeNP{Rudin87}). Let $\gamma$ be a triangle in the open set
$(-\infty,\infty)\times i(-b,a)$; the theorems of Fubini and Cauchy immediately
yield
\begin{align*}
\int_{\partial\gamma} \widehat\varrho(u) \du
 = \int_{\partial\gamma}\int \e^{iux} \varrho(\dx)\du
 = \int \int_{\partial\gamma}\e^{iux}\du\,\varrho(\dx)
 =0,
\end{align*}
as $u\mapsto\e^{iux}$ is analytic for every fixed $x\in\R$. Then, the
analyticity of $\widehat\varrho$ follows from Morera's theorem. For a
justification of the application of Fubini's theorem it is enough to
note that
\begin{align*}
\int\int_{\partial\gamma}\left|\e^{iux}\right| \du\,\varrho(\dx)
 \le \int \int_{\partial\gamma} h(x) \du \, \varrho(dx)
 = \ell(\gamma) \int h(x) \varrho(\dx)
 < \infty,
\end{align*}
where $\ell(\gamma)$ denotes the length of the curve $\partial\gamma$.
\end{proof}

\begin{lemma}\label{rem-expomom}
Let $Y$ be a \lev process and a special semimartingale with $E[Y_t]=0$ for some,
and hence for every, $t>0$. Then
\begin{align*}
E\big[\e^{Y_t^*}\big] \le 8 E\big[\e^{|Y_t|}\big],
\end{align*}
where $Y^*_t=\sup_{0\le s\le t}|Y_s|$.
\end{lemma}
\begin{proof}
Using that $\frac{(Y^*_t)^n}{n!}$ is positive for every $n\ge0$ and the
monotone convergence theorem, we get
\begin{align*}
E\big[\e^{Y_t^*}\big]
 = E \sum_{n=0}^\infty \frac{(Y^*_t)^n}{n!}
 =\sum_{n=0}^\infty E\frac{(Y^*_t)^n}{n!}.
\end{align*}
Now, Remark 25.19 in \citeN{Sato99} yields
\begin{align*}
E(Y^*_t)^n \le 8E|Y_t|^n, \qquad\text{for every }n\ge 1,
\end{align*}
while for $n=0$ the inequality holds trivially. Hence, we get
\begin{equation*}
\sum_{n=0}^\infty E\frac{(Y^*_t)^n}{n!}
 \le 8 \sum_{n=0}^\infty E\frac{|Y_t|^n}{n!}
 =8 E \sum_{n=0}^\infty \frac{|Y_t|^n}{n!}
 =8 E\big[\e^{|Y_t|}\big].  \qedhere
\end{equation*}
\end{proof}

Next, notice that under assumption $(\mathbb{EM})$ we have that
\begin{align*}
\int_\R \left| \e^{Mx}-1-Mx \right| \lambda(\dx) < \infty
 \quad\text{and }
  \int_\R \left| \e^{-Mx}-1+Mx \right| \lambda(\dx) < \infty .
\end{align*}
Let us introduce the following notation:
\begin{align}\label{a-up}
\overline{\alpha}(M)
 := M|b| + \frac{1}{2} cM^2
  + \int_\R \left| \e^{Mx} - 1 - Mx \right| \lambda(\dx)
\end{align}
and
\begin{align}\label{a-down}
\underline{\alpha}(M)
 := M|b|+ \frac{1}{2} cM^2
  + \int_\R \left| \e^{-Mx} - 1 + Mx \right| \lambda(\dx).
\end{align}

\begin{lemma}\label{lem:sec5:neu2}
Let $L=(L_t)_{0\le t\le T}$ be a \lev process that satisfies assumption
$(\mathbb{EM})$. Then we have the following estimates
\begin{align*}
E\big[\e^{u\supL_t}\big]
 \le E\big[\e^{M\supL_t}\big]
 \le 8\mathscr{C}(t,M) < \infty \qquad(u\le M),
\end{align*}
and
\begin{align*}
E\big[\e^{-u\infL_t}\big]
 \le E\big[\e^{-M\infL_t}\big]
 \le 8\mathscr{C}(t,M) < \infty\qquad(u\le M),
\end{align*}
where
$\mathscr{C}(t,M):=\e^{t\overline{\alpha}(M)} + \e^{t\underline{\alpha}(M)}$.
\end{lemma}
\begin{proof}
For $u\le M$  we have
$$\e^{u\supL_t}\le \e^{M\supL_t},$$
since $\supL_t=\sup_{0\le s\le t} L_s$ is nonnegative. Further notice that
\begin{align*}
\supL_t
 = \sup_{0\le s\le t} \big[bs+\sqrt{c}W_s + L^d_s\big]
 \le \sup_{0\le s\le t} \big[\sqrt{c}W_s + L^d_s\big]
   + \sup_{0\le s\le t} [bs],
\end{align*}
where $L_t= bt+\sqrt{c}W_t+L^d_t$ denotes the canonical decomposition of $L$,
with Brownian motion $W$ and a purely discontinuous martingale $L^d=x*(\mu-\nu)$.
Let us further denote by
$$Y_s:= \sqrt{c}W_s + L^d_s.$$
The process $Y$ is not only a martingale but also a  \lev process and a
special semimartingale with local characteristics $(0,c,\lambda)$. We have
\begin{align*}
\supL_t \le \sup_{0\le s\le t} Y_s + |b|t \le Y^*_t + |b|t,
\end{align*}
hence we get that
\begin{align}\label{mom-est-1}
E\big[\e^{M\supL_t}\big]
 \le E\big[\e^{M(Y^*_t+|b|t) }\big]
 = \e^{M|b|t}E\big[\e^{MY^*_t}\big]
 \le 8\e^{M|b|t}E\big[\e^{M|Y_t|}\big],
\end{align}
using Lemma \ref{rem-expomom} for the special semimartingale $Z:=MY$,
which is a \lev process satisfying $E[Z_t]=0$ for every $0\le t\le T$.

Now it is sufficient to notice that
\begin{align}\label{mom-est-2}
E\big[\e^{M|Y_t|}\big]
 \le E\big[\e^{MY_t}\big]  + E\big[\e^{- MY_t}\big],
\end{align}
where Theorem 25.17 in \citeN{Sato99} yields
\begin{align}\label{mom-est-3}
E\big[\e^{MY_t}\big]
 &= \exp\Big(t\frac{cM^2}{2}
             + t\int_\R\left(\e^{Mx}-1-Mx\right)\lambda(\dx)\Big) \nonumber\\
 &\le \e^{(\overline{\alpha}(M) - M|b|)t};
\end{align}
similarly,
\begin{align}\label{mom-est-4}
E\big[\e^{- MY_t}\big]
 &\le \e^{(\underline{\alpha}(M)-M|b|)t}.
\end{align}
Summarizing, we can conclude from \eqref{mom-est-1}--\eqref{mom-est-4} that
\begin{align*}
E\big[\e^{M\supL_t}\big]
 &\le 8 \e^{M|b|t}\Big( \e^{( \overline{\alpha}(M) - M|b|)t}
     + \e^{(\underline{\alpha}(M)-M|b|)t} \Big) \\
 &=8 \Big(\e^{\overline{\alpha}(M)t} + \e^{\underline{\alpha}(M) t} \Big),
\end{align*}
as well as
\begin{align*}
E\big[\e^{-M\infL_t}\big]
 &\le 8 \Big(\e^{\overline{\alpha}(M)t} + \e^{\underline{\alpha}(M) t} \Big). \qedhere
\end{align*}
\end{proof}

A corollary of these results is the existence of an analytic continuation for
the characteristic function $\varphi_{\supL_t}$ of the supremum, resp.
$\varphi_{\infL_t}$ of the infimum, of a \lev process.

\begin{corollary}\label{cor:sec5:neucor1}
Let $L$ be a \lev process that satisfies assumption $(\mathbb{EM})$.
Then, the characteristic function $\varphi_{\overline{L}_t}$ of $\supL_t$,
resp. $\varphi_{\underline{L}_t}$ of $\infL_t$, possesses a continuous
extension
\begin{align*}
\varphi_{\overline{L}_t}(z)
 = \int_\R\e^{izx}P_{\supL_t}(\dx),
 \quad \text{ resp.}\quad
\varphi_{\underline{L}_t}(z)
= \int_\R\e^{izx}P_{\infL_t}(\dx),
\end{align*}
to the half-plane $z\in \{z\in\C:  -M \le \Im z\}$, resp.
$z\in \{z \in\C: \Im z\le M\}$, that is analytic in the interior of the
half-plane $\{z\in\C:  -M<\Im z\}$, resp. $\{z \in\C:  \Im z< M\}$.
\end{corollary}
\begin{proof}
This is a direct consequence of Lemmata \ref{lem:sec5:neu1} and \ref{lem:sec5:neu2}.
\end{proof}

\begin{remark}
One could derive the statement of Corollary \ref{cor:sec5:neucor1} using
the submultiplicativity of the exponential function and Theorem 25.18 in Sato
\citeyear{Sato99}, see Lemma 5 in \citeN{KyprianouSurya05}. However, we
will need the estimates of Lemma \ref{lem:sec5:neu2} in the following sections.
\end{remark}

\subsection{Analytic extension, exponential time case}
The next step is to establish a relationship between the (analytic extension
of the) characteristic function of the supremum, resp. infimum, at a
\emph{fixed} time and at an \emph{independent} and \emph{exponentially
distributed} time. Independent exponential times play a fundamental role in
the fluctuation theory of \lev processes, since they enjoy a property similar
to infinity: the time left after an exponential time is again exponentially
distributed.

Let $\theta$ denote an exponentially distributed random variable with
parameter $q>0$, independent of the \lev process $L$. We denote by
$\supL_\theta$, resp. $\infL_\theta$, the supremum, resp. infimum, process of
$L$ sampled at $\theta$, that is
\begin{align*}
\supL_\theta = \sup_{0\leq u\leq\theta} L_u
 \quad \text{ and }\quad
\infL_\theta = \inf_{0\leq u\leq\theta} L_u.
\end{align*}

\begin{lemma}\label{lem:sec3:3}
Let $L=(L_t)_{0\le t\le T}$ be a \lev process that satisfies assumption
$(\mathbb{EM})$, and let $\theta\sim \mathrm{Exp}(q)$ be independent of
the process $L$.
\par
If $q>\overline{\alpha}(M)\vee\underline{\alpha}(M)$, then the characteristic
function $\varphi_{\overline{L}_{\theta}}$ of $\supL_{\theta}$ possesses a
continuous extension
\begin{align}\label{ch3:eqr-fneu1}
\varphi_{\overline{L}_{\theta}}(z)
 &= \int_\R\e^{izx}P_{\supL_{\theta}}(\dx)
  = q \int_0^\infty \e^{-qt} E\big[\e^{iz\supL_t}\big]\dt
\end{align}
to the half-plane $z\in \{z\in\C:  -M \le \Im z\}$, that is analytic
in the interior of the half-plane $\{z\in\C:  -M<\Im z\}$.
\par
If $q>\overline{\alpha}(M) \vee \underline{\alpha}(M)$, then the characteristic
function $\varphi_{\underline{L}_{\theta}}$ of $\infL_{\theta}$ possesses a
continuous extension
\begin{align}\label{ch3:eqr-fneu2}
\varphi_{\underline{L}_{\theta}}(z)
 &= \int_\R\e^{izx}P_{\infL_{\theta}}(\dx)
  = q \int_0^\infty \e^{-qt} E\big[\e^{iz\infL_t}\big]\dt
\end{align}
to the half-plane  $z\in \{z \in\C: \Im z\le M\}$, that is analytic
in the interior of the half-plane $\{z \in\C: \Im z< M\}$.
\end{lemma}
\begin{proof}
We have that
\begin{align*}
E \big[\e^{u\supL_{\theta}}\big]
 = \int_0^\infty \int_0^\infty \e^{ux} q \e^{-qt} P_{\supL_t}(\dx) \dt
 = \int_0^\infty E\big[\e^{u\supL_t} \big]  q \e^{-qt} \dt,
\end{align*}
and, for $q>\overline{\alpha}(M)\vee\underline{\alpha}(M)$, by Lemma
\ref{lem:sec5:neu2} we get
\begin{align*}
\int_0^\infty E\big[\e^{M\supL_t} \big]  q \e^{-qt} \dt
 \le 8\Bigg( q \int_0^\infty \e^{- t\big[q- \overline{\alpha}(M)\big]}\dt
          + q \int_0^\infty \e^{- t\big[q- \underline{\alpha}(M)\big]}\dt\Bigg)
 < \infty;
\end{align*}
hence, for $u\le M$, we have
\begin{align}\label{hilfmir1}
E\big[\e^{u\supL_{\theta}}\big]
 \le E\big[\e^{M\supL_{\theta}}\big] < \infty
\qquad (q>\overline{\alpha}(M)\vee\underline{\alpha}(M)).
\end{align}
Inequality \eqref{hilfmir1}, together with Lemma \ref{lem:sec5:neu1},
implies that the characteristic function $\varphi_{\supL_{\theta}}$
has a continuous extension to the half-plane
$\{z\in\C:-M\le \Im z\}$, that is analytic in $\{z\in\C:-M<\Im z\}$,
and is given by
$$
\varphi_{\supL_{\theta}}(z) = E\big[\e^{iz\supL_{\theta}}\big],
$$
for every $z\in\C$ with $\Im z \ge -M$. Furthermore Fubini's theorem
yields
$$
E\big[\e^{iz\supL_{\theta}}\big]
 = \int_0^\infty\int_0^\infty \e^{izx} q\e^{-qt} P_{\supL_t}(\dx)\dt
 = q \int_0^\infty \e^{-qt} E\big[\e^{iz\supL_t}\big] \dt.
$$
The application of Fubini's theorem is justified since, for $\Im z\ge -M$
and $q>\overline{\alpha}(M)\vee\underline{\alpha}(M)$, we have
\begin{align*}
E\big[\big|\e^{iz\supL_{\theta}}\big|\big]
 = E\big[\e^{-\Im(z)\supL_{\theta}}\big]
 \le E\big[\e^{M\supL_{\theta}}\big] < \infty
\end{align*}
by inequality \eqref{hilfmir1}. Similarly, we prove the
assertion for the infimum.
\end{proof}

\section{The Wiener--Hopf factorization}
\label{wh-section}

We first provide a statement and brief description of the \wh factorization
of a \lev process, and then show that the \wh factorization holds true for
the analytically extended characteristic functions. Next, we invert the
\wh factorization, and derive an expression for the (analytically extended)
characteristic function of the supremum, resp. infimum, of a \lev process
in terms of the \wh factors.

\subsection{Analyticity}
Fluctuation identities for \lev processes originate from analogous
results for random walks, first derived using combinatorial methods,
see e.g. \citeN{Spitzer64} or \citeN{Feller71}. \citeN{Bingham75}
used this \linebreak
discrete-time skeleton to prove results for \lev processes;
the same approach is followed in the book of \citeN{Sato99}.
\citeANP{GreenwoodPitman80}
(\citeyearNP{GreenwoodPitman80},\citeyearNP{GreenwoodPitman80b})
proved these results for random walks and \lev processes using excursion
theory; see also the books of \citeN{Bertoin96} and \citeN{Kyprianou06}.

The \emph{Wiener--Hopf factorization}\footnote{The historical reasons
leading to the adoption of the terminology ``Wiener--Hopf'' are outlined
in section 6.6 in \citeN{Kyprianou06}.} serves as a common reference to
a multitude of statements in the fluctuation theory for \lev processes,
regarding the distributional decomposition of the excursions of a \lev
process sampled at an independent and exponentially distributed time.
The following statement relates the characteristic function of the supremum,
the infimum, and the \lev process itself. Let $L$ be a \lev process and
$\theta$ an independent, exponentially distributed time with parameter $q$;
then we have that
\begin{align*}
E\big[\e^{iz L_\theta}\big]
 = E\big[\e^{iz\overline L_\theta}\big]E\big[\e^{iz\underline L_\theta}\big]
\end{align*}
or equivalently,
\begin{align*}
\frac{q}{q-\kappa(iz)}=\varphi^+_q(z)\varphi^-_q(z), \qquad z\in\R;
\end{align*}
here $\kappa$ denotes the cumulant generating function of $L_1$, cf.
\eqref{levy-cmu}, and $\varphi^+_q$, $\varphi^-_q$ denote the so-called
Wiener--Hopf factors.

In the sequel, we will make use of the Wiener--Hopf factorization as stated in
the beautiful book of \citeN{Kyprianou06}, and prove the analytic extension of
the Wiener--Hopf factors to the open half-plane $\{z\in\C: \Im z>-M\}$.

Recall the definitions of \eqref{a-up} and \eqref{a-down}, and let us
denote by
\begin{align*}
\alpha^*(M)
 := \max\Big\{\overline{\alpha}(M),\underline{\alpha}(M)\Big\}.
\end{align*}

\begin{theorem}[Wiener--Hopf factorization]\label{thm:wh}
Let $L$ be a \lev process  that satisfies assumption $(\mathbb{EM})$
(and is not a compound Poisson process). The Laplace transform of
$\supL_{\theta}$, resp. $\infL_{\theta}$, at an independent and exponentially
distributed time $\theta$, $\theta\sim \mathrm{Exp}(q)$, with $q>\alpha^*(M)$,
can be identified from the Wiener--Hopf factorization of $L$ via
\begin{align}\label{WH-sup}
E\big[\e^{-\beta\supL_{\theta}}\big]
 &= \int_0^\infty q E\big[\e^{-\beta\supL_t}\big] \e^{-qt} \dt
  = \frac{\overline\kappa(q,0)}{\overline\kappa(q,\beta)}
\end{align}
and
\begin{align}\label{WH-inf}
E\big[\e^{\beta\infL_{\theta}}\big]
 &=  \int_0^\infty q E\big[\e^{\beta\infL_t}\big] \e^{-qt} \dt
  = \frac{\underline{\kappa}(q,0)}{\underline\kappa(q,\beta)}
\end{align}
for $\beta \in \{\beta\in\C: \Re(\beta)>-M\}$.
The Laplace exponent of the ascending, resp. descending, ladder
process $\overline\kappa(\alpha,\beta)$, resp.
$\underline\kappa(\alpha,\beta)$, for $\alpha\geq \alpha^*(M)$ and
$\overline{k},\underline{k}>0$, has an analytic extension to
$\beta\in \{\beta\in\C: \Re(\beta)>-M\}$ and is given by
\begin{align}\label{kappa}
\overline\kappa(\alpha,\beta)
 &= \overline{k}
    \exp \Bigg(\int_0^\infty\int_{(0,\infty)}(\e^{-t}-\e^{-\alpha t-\beta x})\frac1t P_{L_t}(\dx)\dt\Bigg),
\end{align}
and
\begin{align}\label{kappa-hat}
\underline\kappa(\alpha,\beta)
 &= \underline{k}
    \exp\Bigg(\int_0^\infty\int_{(-\infty,0)}(\e^{-t}-\e^{-\alpha t+\beta x})\frac1tP_{L_t}(\dx)\dt\Bigg).
\end{align}
\end{theorem}

\begin{remark}
Note that the Wiener--Hopf factors $\varphi^+_q$ and $\varphi^-_q$ are related
to the Laplace exponents of the ascending and descending ladder process
$\overline\kappa$ and $\underline\kappa$ via
\begin{align}
\varphi^+_q(i\beta) = \frac{\overline\kappa(q,0)}{\overline\kappa(q,\beta)}
 \quad \text{ and } \quad
\varphi^-_q(-i\beta) = \frac{\underline\kappa(q,0)}{\underline\kappa(q,\beta)}.
\end{align}
\end{remark}

We will prepare the proof of this theorem with an intermediate lemma. Let us
denote the positive part by $a_+:=\max\{a,0\}$.

\begin{lemma}\label{lem-neu-phi1}
Let $L$ be a \lev process that satisfies assumption $(\mathbb{EM})$.
For $q>\kappa(M)_+$ the maps
\begin{align}\label{an-map}
z\mapsto \int_0^\infty \int_{(0,\infty)}
         \big(1 - \e^{izx} \big) P_{L_t}(\dx) \frac{\e^{-qt}}{t} \dt
\end{align}
and
\begin{align}
z\mapsto \int_0^\infty \int_{(0,\infty)}
         \big(\e^{-t} - \e^{-qt+izx} \big) P_{L_t}(\dx) \frac{1}{t} \dt
\end{align}
are well defined and analytic in the open half plane
$\big\{z\in\C: \Im(z)>-M\big \}$.
\end{lemma}
\begin{proof}
We will show that for every compact subset $K\subset\{z\in\C:\,\Im(z)>-M\}$,
there is a constant $C=C(K)>0$ such that
\begin{align}\label{Anker}
\int_0^\infty \int_{(0,\infty)}\big| \e^{izx}-1\big| P_{L_t}(\dx) \frac{\e^{-qt}}{t} \dt <C(K),
\end{align}
for every $z\in K$. Then, applying Lebesgue's dominated convergence theorem
yields the continuity of the function
\begin{align*}
z\mapsto \int_0^\infty\int_{(0,\infty)} \big( \e^{izx} - 1\big) P_{L_t}(\dx) \frac{\e^{-qt}}{t} \dt
\end{align*}
inside the half-plane $\{z\in\C:\, \Im(z)>-M\}$. Moreover, let $\gamma$ be
an arbitrary triangle inside $\{z\in\C:\,\Im(z)>-M\}$; the theorems of Fubini
and Cauchy yield
\begin{multline}
\int_{\partial\gamma}\int_0^\infty\int_{(0,\infty)}
          \big(\e^{izx}-1\big) P_{L_t}(\dx) \frac{\e^{-qt}}{t}\dt\,\dz\\
  = \int_0^\infty\int_{(0,\infty)}\int_{\partial\gamma} \big(\e^{izx}-1\big)\dz\, P_{L_t}(\dx)\frac{\e^{-qt}}{t}\dt
  = 0 \,.
\end{multline}
Hence, applying Morera's theorem yields the analyticity of \eqref{an-map}
in the open half-plane $\{z\in\C:\, \Im(z)>-M\}$.

The assertion for the second map immediately follows from the identity
\begin{align*}
\big(\e^{-t} - \e^{-qt+izx}\big) t^{-1}
 = \big(1- \e^{izx} \big)\e^{-qt}  t^{-1} + \big(\e^{-t} - \e^{-qt}\big) t^{-1}
\end{align*}
and the integrability of the second part, since
\begin{align*}
\int_\epsilon^\infty\big|\e^{-t} - \e^{-qt}\big| t^{-1} \dt <\infty
\end{align*}
and
\begin{align*}
\int_0^\epsilon\big|\e^{-t} - \e^{-qt}\big| t^{-1} \dt
 = \int_0^\epsilon\big|\e^{t(q-1)} - 1\big| \e^{-qt}t^{-1} \dt
 \le C |q-1|\int_0^\epsilon \e^{-qt}\dt
 <\infty,
\end{align*}
with $C>1$, for $\epsilon>0$ small enough.

To show estimation \eqref{Anker}, we choose a constant $k=k(K)>0$ only
depending on the compact set $K$, such that $|z|<k$ for every $z\in K$, and we
write
\begin{align}\label{hilflemma59}
\lefteqn{\int_{(0,\infty)} \big| \e^{izx} - 1\big| P_{L_t}(\dx)}\nonumber\\
 &= \int_{(0,1/k]} \big| \e^{izx} - 1\big| P_{L_t}(\dx)
  + \int_{(1/k,\infty)} \big| \e^{izx} - 1\big| P_{L_t}(\dx)\nonumber\\
 &\le \int_{(0,1/k]} |zx| P_{L_t}(\dx) + \int_{(1/k,\infty)} \big| \e^{izx}\big| P_{L_t}(\dx)
  + \int_{(1/k,\infty)}P_{L_t}(\dx)\,.
\end{align}
Using inequality (30.13) of Lemma 30.3 in \citeN{Sato99} we can deduce
\begin{align*}
\int_{(0,1/k]} |zx| P_{L_t}(\dx)
\le k  \int_{(0,1/k]} |x| P_{L_t}(\dx)
\le kE \big[ |L_t| 1_{\{|L_t| \le 1/k\}}\big]
\le C_1(K) t^{1/2}\,,
\end{align*}
with a constant $C_1(K)$ that depends only on the compact set $K$.
Similarly, using inequality (30.10) in \citeN{Sato99}, we can
estimate the last term of \eqref{hilflemma59}
\begin{align*}
\int_{(1/k,\infty)}P_{L_t}(\dx)
 = P\big(\{ L_t > 1/k \} \big) \le P\big(\{ |L_t| > 1/k \} \big)
 \le C_2(K) t \,,
\end{align*}
with a constant $C_2(K)$ that depends only on the compact set $K$.
In order to estimate the second term of inequality \eqref{hilflemma59},
let us note that we may choose $\epsilon>0$ small enough, such that
for every $z\in K$, we have $-\Im(z)<M'<M$ with $M':=M(1-\epsilon)$,
and we get
\begin{align*}
 \int_{(1/k,\infty)} \big|\e^{izx}\big| P_{L_t}(\dx)
\le  E\big[\e^{M'L_t} 1_{\{|L_t| >1/k\}} \big] .
\end{align*}
Applying H\"older's inequality with $p:=\frac{1}{1-\epsilon}$ and
$q:=\frac{1}{\epsilon}$, together with Lemma 30.3 in \citeN{Sato99},
yields
\begin{align*}
E\big[\e^{M'L_t} 1_{\{|L_t| >1/k\}} \big]
 &\le \Big( E\big[\e^{p M'L_t}\big] \Big)^{1/p} \Big( P\big(\{|L_t| >1/k\}\big) \Big)^{1/q} \\
 &\le C_3(K) t^\epsilon\e^{(1-\epsilon) \kappa(M) t}\,.
\end{align*}

Altogether we have
\begin{align*}
\int_{(0,\infty)} \big|\e^{izx} - 1\big| P_{L_t}(\dx)
 \le C_1(K) t^{1/2} + C_2(K) t + C_3(K) t^\epsilon \e^{(1-\epsilon) \kappa(M) t},
\end{align*}
with positive constants $C_1(K),\, C_2(K)$ and $C_3(K)$ that only depend on the
compact set $K$. As $q>(1-\epsilon)(\kappa(M))_+$, we can conclude \eqref{Anker},
which completes the proof.
\end{proof}

\begin{proof}[Proof of Theorem \ref{thm:wh}]
For $\beta\in\C$ with $\Re\beta\ge0$ the assertion follows directly from
Theorem 6.16 (ii) and (iii) in \citeN{Kyprianou06}.

From Lemma \ref{lem:sec3:3} we know that for $q>\alpha^*(M)$ the function
\begin{align*}
\beta\mapsto \varphi_{\supL_{\theta}}(i\beta) = E\big[\e^{-\beta \supL_{\theta}}\big]
\end{align*}
has an analytic extension to the half-plane
\begin{align*}
\{\beta \in\C: \Re(\beta)> -M\},
\end{align*}
whereas Lemma \ref{lem-neu-phi1} yields that if $q>\alpha^*(M)$, the mapping
\begin{align*}
\beta\mapsto \frac{\overline\kappa(q,0)}{\overline\kappa(q,\beta)}
\end{align*}
has an analytic extension to the half-plane
\begin{align*}
\{\beta \in\C: \Re(\beta)> -M\},
\end{align*}
while identity \eqref{kappa} still holds for this extension.
The identity theorem for holomorphic functions yields that equation
\eqref{WH-sup} holds for every $\{\beta \in\C: \Re(\beta)> -M\}$
if $q>\alpha^*(M)$. The proof for equations \eqref{WH-inf} and
\eqref{kappa-hat} follows along the same lines.
\end{proof}

\begin{remark}
Note that, by analogous arguments, we can prove that the Laplace exponent
of the ascending, resp. descending, ladder process $\overline\kappa(\alpha,\beta)$,
resp. $\underline\kappa(\alpha,\beta)$, has an analytic extension to
$\alpha\in \{\alpha\in\C: \Re(\alpha)>\alpha^*(M)\}$, which is given by
\eqref{kappa}, resp. \eqref{kappa-hat}.
\end{remark}

\subsection{Inversion}
The next step is to invert the Laplace transform in the \wh factorization in
order to recover the characteristic function of $\supL_t$, at a \textit{fixed}
time $t$. Let us mention that although the Wiener--Hopf factorization and the
characteristic function of $\supL_\theta$ are discussed in several textbooks,
the extended characteristic function of $\supL_t$ at a fixed time has not been
studied in the literature before.

The main result is Theorem \ref{inv:wh}, which will make use of the following
auxiliary lemma.

\begin{lemma}\label{map}
The maps $t\mapsto E\big[\e^{-\beta\supL_t}\big]$ and $t\mapsto
E\big[\e^{\beta\infL_t}\big]$ are continuous for all $\beta\in\C$
with $\Re\beta\in[-M,\infty)$.
\end{lemma}
\begin{proof}
Since the \lev process $L$ is right continuous, stochastically continuous
and $\supL$ is an increasing process, we get that $\supL_s\nearrow\supL_t$ a.s.
as $s\rightarrow t$.

As $\supL_s \ge0$ we have
\begin{align*}
\big|\e^{-\beta \supL_s}\big|
 = \e^{-\Re(\beta) \supL_s}
 \le \e^{M\supL_s}
 \le \e^{M\supL_t},
\end{align*}
and we may apply the dominated convergence theorem to get
\begin{align*}
E\big[\e^{-\beta \supL_s}\big]
 \rightarrow E\big[\e^{-\beta \supL_t}\big] \qquad \text{as }s\rightarrow t,
\end{align*}
for every $\beta\in\C$ with $\Re(\beta)\ge -M$.
Analogously, taking into account that
$\big|\e^{\beta\infL_s}\big|\le\e^{-M\infL_s}$ for $\Re\beta\ge-M$,
the dominated convergence theorem yields the continuity of the second
map.
\end{proof}

\begin{theorem}\label{inv:wh}
Let $L$ be a \lev process that satisfies assumption $(\mathbb{EM})$ (and
is not a compound Poisson process). The  Laplace transform of
$\supL_t$ and $\infL_t$ at a fixed time $t$, $t\in[0,T]$, is given by
\begin{align}
E\big[\e^{-\beta\supL_t}\big]
 &= \lim_{A\to\infty} \frac{1}{2\pi} \int_{-A}^A
    \frac{\e^{t(Y+iv)}}{Y+iv} \frac{\overline\kappa(Y+iv,0)}{\overline\kappa(Y+iv,\beta)}\ud v,
\end{align}
and
\begin{align}
E\big[\e^{\beta\infL_t}\big]
 &= \lim_{A\to\infty} \frac{1}{2\pi} \int_{-A}^A
    \frac{\e^{t(\widetilde{Y}+iv)}}{\widetilde{Y}+iv}
    \frac{\underline\kappa(\widetilde{Y}+iv,0)}{\underline\kappa(\widetilde{Y}+iv,-\beta)}\ud v,
\end{align}
for $\beta\in\C$ with $\Re\beta\in(-M,\infty)$
and $Y,\widetilde{Y}>\alpha^*(M)$.
\end{theorem}
\begin{proof}
Theorem \ref{thm:wh}, together with equation \eqref{WH-sup},
immediately yield
\begin{align}
\int_0^\infty \e^{-qt}E\big[\e^{-\beta \supL_t}\big]\dt
= \frac{1}{q}\frac{\overline\kappa(q,0)}{\overline\kappa(q,\beta)},
\end{align}
for $\beta\in\C$ with $\Re(\beta)>-M$ and $q>\alpha^*(M)$.

In order to deduce that we can invert this Laplace transform, we want to
verify the assumptions of Satz 4.4.3 in \citeN{Doetsch50} for the real
and imaginary part of $t\mapsto E\big[\e^{-\beta\supL_t}\big]$.
From the proof of Lemma \ref{lem:sec3:3} we get that
\begin{align*}
\int_0^\infty\e^{-qt}\Big|E\big[\e^{-\beta \supL_t}\big]\Big|\dt
 \le \int_0^\infty\e^{-qt} E\big[\e^{-\Re(\beta)\supL_t}\big]\dt<\infty;
\end{align*}
this yields the required integrability, i.e. absolute convergence, of
\begin{align*}
 \int_0^\infty\e^{-qt} \Big|\Im\big(E\big[\e^{\beta\supL_t}\big]\big)\big|\dt
\quad\text{ and }\quad
 \int_0^\infty\e^{-qt} \Big|\Re\big(E\big[\e^{\beta\supL_t}\big]\big)\big|\dt,
\end{align*}
for $q>\alpha^*(M)$.
Further the real and imaginary part of
$t\mapsto E\big[\e^{-\beta \supL_t}\big]$
are of bounded variation for $\beta\in\C$ with
$\Re\beta\in(-M,\infty)$.

Let us verify this assertion for the imaginary part, for
$-M<\Re(\beta)\le0$ and $\Im(\beta) \le 0$. We have that
\begin{align*}
\Im\left(E\big[\e^{-\beta \supL_t}\big]\right)
 = i E\big[\sin\big(-\Im(\beta)\supL_t\big)\e^{-\Re(\beta) \supL_t}\big].
\end{align*}
We can decompose $\sin(x)=f(x)-g(x)$, where $f$ and $g$ are increasing functions
with $f(0)=g(0)=0$, and $|f(x)|\le x$ and $|g(x)|\le x$. It follows that
\begin{align*}
\sin\big(-\Im(\beta)\supL_t\big)\e^{-\Re(\beta) \supL_t}
 = f\big(-\Im(\beta)\supL_t\big)\e^{-\Re(\beta) \supL_t}
 - g\big(-\Im(\beta)\supL_t\big)\e^{-\Re(\beta) \supL_t},
\end{align*}
where both terms are increasing in time and are integrable, since
\begin{align*}
E\Big[\big|h\big(-\Im(\beta)\supL_t\big)\e^{-\Re(\beta) \supL_t}\big|\Big]
 &\le \big| \Im(\beta) \big|
      E\Big[ \big|\supL_t\big| \e^{-\Re(\beta) \supL_t} \Big]\\
 &\le \text{const}\cdot E\big[ \e^{M \supL_t} \big]
  < \infty,
\end{align*}
for $h=g$ and $h=f$. The assertion for the other parts follows similarly.

Now, using the continuity of the map $t\mapsto E\big[\e^{-\beta\supL_t}\big]$,
cf. Lemma \ref{map}, we may apply Satz 4.4.3 in \citeN{Doetsch50}, to invert
this Laplace transform; that is, to conclude that
\begin{align}
E\big[\e^{-\beta\supL_t}\big]
&= \text{(p.v.)}\, \frac{1}{2\pi i}\int_{Y-i\infty}^{Y+i\infty}
     \frac{\e^{tz}}{z} \frac{\overline\kappa(z,0)}{\overline\kappa(z,\beta)}\ud z \nonumber\\
&= \lim_{A\to\infty} \frac{1}{2\pi} \int_{-A}^A
     \frac{\e^{t(Y+iv)}}{Y+iv} \frac{\overline\kappa(Y+iv,0)}{\overline\kappa(Y+iv,\beta)}\ud v,
\end{align}
for all $\beta\in\C$ with $\Re\beta\in(-M,\infty)$ and for every
$Y>\alpha^*(M)$. The proof for the infimum follows along the same lines.
\end{proof}

\section{\lev processes: examples and  properties}
\label{Levy-ex-prop}

We first state some conditions for the continuity of the law of a \lev process,
and the continuity of the law of the supremum of a \lev process. Then, we describe
the most popular \lev models for financial applications, and comment on their
path and moment properties which are relevant for the application of Fourier
transform valuation formulas.

\subsection{Continuity properties}
The valuation theorem for discontinuous payoff functions (Theorem 2.7 in EGP),
and the analysis of the properties of discontinuous payoff functions (Examples
5.2, 5.3 and 5.4 in EGP), show that if the measure of the underlying random
variable does not have atoms, then the valuation formula is valid as a pointwise limit.
Thus, we present sufficient conditions for the continuity of the law of a \lev
process and its supremum, and discuss these conditions for certain popular
examples.

\begin{statement}\label{cont-Levy}
Let $L$ be a \lev process  with triplet ($b,c,\lambda$). Then, Theorem
27.4 in \citeN{Sato99} yields that the law $P_{L_t}$, $t\in[0,T]$, is
\textit{atomless} iff $L$ is a process of \textit{infinite variation}
or \textit{infinite activity}. In other words, if one of the following
conditions holds true:
\begin{description}
\item[(a)] $c\neq0$ or $\int_{\{|x|\le1\}}|x|\lambda(\dx)=\infty$;
\item[(b)] $c=0$, $\lambda(\R)=\infty$ and $\int_{\{|x|\le1\}}|x|\lambda(\dx)<\infty$.
\end{description}
\end{statement}

\begin{statement}\label{cont-sup-Levy}
Let $L$ be a \lev process and assume that
\begin{description}
\item[(a)] $L$ has \textit{infinite variation}, or
\item[(b)] $L$ has \textit{infinite activity} and is \emph{regular upwards}.
           Regular upwards means that
           $P(\tau_0=0)=1$ where $\tau_0:=\inf\{t>0: L_t(\omega)>0\}$.
\end{description}
Then, Lemma 49.3 in \citeN{Sato99} yields that $\overline{L}_t$ has a
\textit{continuous} distribution for every $t\in[0,T]$. The statement
for the infimum of a \lev process is analogous.
\end{statement}

\subsection{Examples}
Next, we describe the most popular \lev processes for applications in
mathematical finance, namely the generalized hyperbolic (GH) process,
the CGMY process and the Meixner process. We present their characteristic
functions, which are essential for the application of Fourier transform
methods for option pricing, and its domain of definition. We also discuss
their path properties which are relevant for option pricing. For an
interesting survey on the path properties of \lev processes we refer to
Kyprianou and Loeffen \citeyear{KyprianouLoeffen05}.

\begin{example}[GH model]
Let $H=(H_t)_{0\le t\le T}$ be a generalized hyperbolic process with
$\mathcal L(H_1)= \mathrm{GH}(\lambda,\alpha,\beta,\delta,\mu)$, cf.
\citeN[p.~321]{Eberlein01a} or \citeN{EberleinPrause02}. The characteristic
function of $H_1$ is
\begin{align}
\varphi_{H_1}(u) = \e^{iu\mu}
  \bigg(\frac{\alpha^2-\beta^2}{\alpha^2-(\beta+iu)^2}\bigg)^{\frac{\lambda}{2}}
  \frac{K_{\lambda}\big(\delta\sqrt{\alpha^2-(\beta+iu)^2}\big)}
  {K_{\lambda}\big(\delta\sqrt{\alpha^2-\beta^2}\big)},
\end{align}
where $K_{\lambda}$ denotes the Bessel function of the third kind with
index $\lambda$ (cf. \citeNP{AbramowitzStegun68}); the moment generating
function exists for $u\in(-\alpha-\beta,\alpha-\beta)$. The sample paths
of a generalized hyperbolic \lev process have infinite variation. Thus,
by Statements \ref{cont-Levy} and \ref{cont-sup-Levy}, we can deduce that
the laws of both a GH \lev process and its supremum do not have atoms.

The class of generalized hyperbolic distributions is not closed
under convolution, hence the distribution of $H_t$ is no longer a
generalized hyperbolic one. Nevertheless, the characteristic function
of $\mathcal L(H_t)$ is given explicitly by
\begin{displaymath}
 \varphi_{H_t}(u) = \left(\varphi_{H_1}(u)\right)^t.
\end{displaymath}

A class closed under certain convolutions is the class of normal inverse Gaussian
distributions, where $\lambda=-\frac12$; cf. \citeN{Barndorff-Nielsen98}.
In that case, $\mathcal L(H_t)= \mathrm{NIG}(\alpha,\beta,\delta t,\mu t)$
and the characteristic function resumes the form
\begin{align}
\varphi_{H_t}(u) = e^{iu\mu t}
  \frac{\exp(\delta t\sqrt{\alpha^2-\beta^2})}{\exp(\delta t\sqrt{\alpha^2-(\beta+iu)^2})}.
\end{align}
Another interesting subclass is given by the hyperbolic
distributions which arise for $\lambda=1$; the hyperbolic model has
been introduced to finance by \citeN{EberleinKeller95}.
\end{example}

\begin{example}[CGMY model]
Let $H=(H_t)_{0\le t\le T}$ be a CGMY \lev process, cf.
\citeN{Carretal02}; another name for this process is (generalized)
tempered stable process (see e.g. \citeANP{ContTankov03} \citeyear{ContTankov03}). The \lev
measure of this process has the form
\begin{align*}
\lambda^{CGMY}(\dx) = C\frac{\e^{-Mx}}{x^{1+Y}}1_{\{x>0\}}\dx
           + C\frac{\e^{Gx}}{|x|^{1+Y}}1_{\{x<0\}}\dx,
\end{align*}
where the parameter space is $C,G,M>0$ and $Y\in(-\infty,2)$.
Moreover, the characteristic function of $H_t$, $t\in[0,T]$, is
\begin{align}
\varphi_{H_t}(u)
 = \exp\Big( t\,C\,\Gamma(-Y)\big[(M-iu)^Y+(G+iu)^Y-M^Y-G^Y\big] \Big),
\end{align}
for $Y\neq0$, and the moment generating function exists for
$u\in[-G,M]$.

The sample paths of the CGMY process have unbounded variation if
$Y\in[1,2)$, bounded variation if $Y\in(0,1)$, and are of compound
Poisson type if $Y<0$. Moreover, the CGMY process is regular upwards
if $Y>0$; cf. \citeN{KyprianouLoeffen05}. Hence, by Statements
\ref{cont-Levy} and \ref{cont-sup-Levy}, the laws of a CGMY \lev process,
and its supremum, do not have atoms if $Y\in(0,2)$.

The CGMY process contains the Variance Gamma process (cf. Madan and Seneta
\citeyear{MadanSeneta90}) as a subclass, for $Y=0$. The characteristic
function of $H_t$, $t\in[0,T]$, is
\begin{align}
\varphi_{H_t}(u)
 = \exp\bigg( t\,C\,\Big[-\log\Big(1-\frac{iu}M\Big)-\log\Big(1+\frac{iu}G\Big)\Big] \bigg),
\end{align}
and the moment generating function exists for $u\in[-G,M]$.
The paths of the variance gamma process have bounded variation, infinite
activity and are regular upwards. Thus, the laws of a VG \lev process
and its supremum do not have atoms.
\end{example}

\begin{example}[Meixner model]
Let $H=(H_t)_{0\le t\le T}$ be a Meixner process with
$\mathcal L(H_1)=\Meixner(\alpha,\beta,\delta)$, $\alpha>0$,
$-\pi<\beta<\pi$, $\delta>0$, cf. \citeN{SchoutensTeugels98} and
\citeN{Schoutens02}. The characteristic function of $H_t$, $t\in[0,T]$, is
\begin{align}
\varphi_{H_t}(u)
 = \left(\frac{\cos\frac\beta2}{\cosh\frac{\alpha u-i\beta}{2}}\right)^{2\delta t},
\end{align}
and the moment generating function exists for
$u\in\big(\frac{\beta-\pi}{\alpha},\frac{\beta+\pi}{\alpha}\big)$.
The paths of a Meixner process have infinite variation. Hence the laws of a
Meixner \lev process and its supremum do not have atoms.
\end{example}

\section{Applications in finance}
\label{ch3:sc5}

In this section, we derive valuation formulas for lookback options, one-touch
options and equity default swaps, in models driven by \lev processes. We combine
the results on the \wh factorization and the characteristic function of the
supremum of a \lev process from this paper, with the results on Fourier transform
valuation formulas derived in EGP. Note that the results presented in the sequel
are valid for all the examples discussed in section \ref{Levy-ex-prop}.

We model the price process of a financial asset $S=(S_t)_{0\le t\le T}$ as an
exponential \lev process, i.e. a stochastic process with representation
\begin{equation}\label{asset-price}
S_t = S_0 \e^{L_t},
 \qquad 0\le t\le T
\end{equation}
(shortly: $S=S_0\,\e^L$). Every \lev process $L$, subject to Assumption \EM,
has the canonical decomposition
\begin{equation}\label{asset-candec}
 L_t = bt + \sqrt{c}W_t + \int_0^t\int_{\R}x(\mu-\nu)\dsdx,
\end{equation}
where \prozess[W] denotes a $P$-standard Brownian motion and $\mu$ denotes the
random measure associated with the jumps of $L$; cf. Jacod and Shiryaev
\citeyear[Chapter II]{JacodShiryaev03}.

Let $\mathcal M(P)$ denote the class of martingales on the stochastic basis
$\mathscr B$. The martingale condition for an asset $S$ is
\begin{equation}\label{asset drift}
 S = S_0\,\e^L\in\mathcal{M}(P)
  \Leftrightarrow
 b + \frac{c}{2} + \int_\R(\e^x-1-x)\lambda(\dx) = 0;
\end{equation}
cf. \shortciteN{EberleinPapapantoleonShiryaev06} for the details. That is, throughout
the rest of this paper, we will assume that $P$ is a \emph{martingale measure} for $S$.

\subsection{Lookback options}
The results on the characteristic function of the supremum of a \lev process,
cf. section \ref{wh-section}, allow us to price lookback options in models
driven by \lev processes using Fourier methods. Excluded are only compound
Poisson processes. Assuming that the asset price evolves as an exponential
\lev process, a fixed strike lookback call option with payoff
\begin{align}\label{lb-pay}
(\overline S_T - K)^+=(S_0\e^{\supL_T}-K)^+
\end{align}
can be viewed as a call option where the driving process is the \textit{supremum}
of the underlying \lev processes $L$. Therefore, the price of a lookback call
option is provided by the following result.

\begin{theorem}\label{lb-value}
Let $L$ be a \lev process that satisfies Assumption \EM. The price of
a fixed strike lookback call option with payoff \eqref{lb-pay} is given
by
\begin{align}\label{ch3:lc}
\mathbb{C}_T(\overline S;K)
 = \frac{1}{2\pi}\int_\R S_0^{R-iu}
   \varphi_{\supL_T}(-u-iR)\frac{K^{1+iu-R}}{(iu-R)(1+iu-R)}\ud u,
\end{align}
where
\begin{align}\label{supL-cf}
\varphi_{\supL_T}(-u-iR)
 = \lim_{A\to\infty} \frac{1}{2\pi} \int_{-A}^A
   \frac{\e^{T(Y+iv)}}{Y+iv}
   \frac{\overline\kappa(Y+iv,0)}{\overline\kappa(Y+iv,iu-R)}\ud v,
\end{align}
for $R\in(1,M)$ and $Y>\alpha^*(M)$.
\end{theorem}
\begin{proof}
We aim at applying Theorem 2.2 in EGP, hence we must check if conditions
(C1)--(C3) (of EGP) are satisfied.
Assumption \EM, coupled with Corollary \ref{cor:sec5:neucor1}, yields
that $M_{\supL_T}(R)$ exists for $R\in(-\infty,M)$, hence condition (C2)
is satisfied. Now, the Fourier transform of the payoff function
$f(x)=(\e^x-K)^+$ is
\begin{align*}
\widehat{f}(u+iR) = \frac{K^{1+iu-R}}{(iu-R)(1+iu-R)},
\end{align*}
and conditions (C1) and (C3) are satisfied for $R\in(1,\infty)$; cf.
Example 5.1 in EGP. Further,
the extended characteristic function $\varphi_{\supL_T}$ of $\supL_T$ is
provided by Theorem \ref{inv:wh} and equals \eqref{supL-cf} for
$R\in(-\infty,M)$ and $Y>\alpha^*(M)$. Finally,  Theorem
2.2 in EGP delivers the asserted valuation formula \eqref{ch3:lc}.
\end{proof}

\begin{remark}
Completely analogous formulas can be derived for the fixed strike lookback
put option with payoff $(K-\underline S_T)^+$ using the results for the
infimum of a \lev process. Moreover, floating strike lookback options can
be treated by the same formulas making use of the duality relationships
proved in \citeN{EberleinPapapantoleon05} and
Eberlein, Papapantoleon, and Shiryaev
\citeyear{EberleinPapapantoleonShiryaev06}.
\end{remark}

\subsection{One-touch options}
Analogously, we can derive valuation formulas for one-touch options in assets
driven by \lev processes using Fourier transform methods; here, the exceptions
are compound Poisson processes and non-regular upwards, finite variation, \lev
processes. Assuming that the asset price evolves as an exponential \lev process,
a one-touch call option with payoff
\begin{align}\label{ot-pay}
1_{\big\{\overline S_T>B\big\}}
 = 1_{\big\{\supL_T>\log(\frac{B}{S_0})\big\}}
\end{align}
can be valued as a digital call option where the driving process is
the supremum of the underlying \lev process.

\begin{theorem}\label{ot-value}
Let $L$ be a \lev process with infinite variation, or a regular upwards process with
infinite activity, that satisfies Assumption \EM. The
price of a one-touch option with payoff \eqref{ot-pay} is given by
\begin{align}\label{ch3:ot}
\mathbb{DC}_T(\overline S;B)
 &= \lim_{A\rightarrow\infty}\frac{1}{2\pi}\int_{-A}^A
    S_0^{R+iu}\varphi_{\supL_T}(u-iR)\frac{B^{-R-iu}}{R+iu}\ud u\\ \nonumber
 &= P\big(\supL_T>\log(B/S_0)\big),
\end{align}
for $R\in(0,M)$ and $Y>\alpha^*(M)$, where
$\varphi_{\supL_T}$ is given by \eqref{supL-cf}.
\end{theorem}
\begin{proof}
We will apply Theorem 2.7 in EGP, hence we must check conditions (D1)--(D2).
As in the proof of Theorem \ref{lb-value}, Assumption \EM shows that condition
(D2) is satisfied for $R\in(-\infty,M)$, while Theorem \ref{inv:wh} provides
the characteristic function of $\supL_T$, given by \eqref{supL-cf}. Example 5.2
in EGP yields that the Fourier transform of the payoff function
$f(x)=1_{\{x>\log B\}}$ equals
\begin{align}
 \widehat{f}(iR-u) = \frac{B^{-R-iu}}{R+iu},
\end{align}
and condition (D1) is satisfied for $R\in(0,\infty)$. In addition, if the
measure $P_{\supL_T}$ is atomless, then the valuation function is continuous
and has bounded variation. Now, by Statement \ref{cont-sup-Levy}, we know
that the measure $P_{\supL_T}$ is atomless exactly when $L$ has infinite
variation, or has infinite activity and is regular upwards. Therefore,
Theorem 2.7 in EGP applies, and results in the valuation formula
\eqref{ch3:ot} for the one-touch call option.
\end{proof}

\begin{remark}
Completely analogous valuation formulas can be derived for the digital put
option with payoff $1_{\{\underline{S}_T<B\}}$.
\end{remark}

\begin{remark}
Summarizing the results of this paper and of EGP, when dealing with
\textit{continuous} payoff functions the valuation formulas can be applied
to \textit{all} \lev processes. When dealing with \textit{discontinuous}
payoff functions, then the valuation formulas apply to most \lev processes
\textit{apart} from \textit{compound Poisson} type processes without diffusion
component, and finite variation \lev processes which are not \textit{regular upwards}.
This is true for both non-path-dependent as well as for \textit{path-dependent
exotic} options.
\end{remark}

\begin{remark}
Arguing analogously to Theorems \ref{lb-value} and \ref{ot-value}, we  can derive
the price of options with a ``general'' payoff function $f(\overline{L}_T)$. For
example, one could consider payoffs of the form $[(\overline{S}_T-K)^+]^2$ or
$\overline{S}_T1_{\{\overline{S}_T>B\}}$; cf. \citeANP{Raible00}
\citeyear[Table 3.1]{Raible00} and 
Example 5.3 in EGP for the corresponding Fourier transforms.
\end{remark}

\subsection{Equity default swaps}
Equity default swaps were recently introduced in financial markets, and offer
a link between equity and credit risk. The structure of an equity default swap
imitates that of a credit default swap: the protection buyer pays a fixed
premium in exchange for an insurance payment in case of `default'. In this
case `default', also called the `equity event', is defined as the first time
the asset price process drops below a fixed barrier, typically 30\% or 50\% of
the initial value $S_0$.

Let us denote by $\tau_B$ the first passage time below the barrier level $B$,
i.e.
\begin{align*}
\tau_B = \inf\{t\geq0; S_t\leq B\}.
\end{align*}
The protection buyer pays a fixed premium denoted by $\mathcal K$ at
the dates $T_1, T_2,\dots, T_N=T$, provided that default has not
occurred, i.e. $T_i<\tau_B$. In case of default, the protection
seller makes the insurance payment $\mathcal C$, which is typically
50\% of the initial value. The premium $\mathcal K$ is fixed
such that the value of the equity default swap at inception is zero,
hence we get
\begin{align}\label{EDS}
\mathcal K = \frac{\mathcal C E\big[\e^{-r\tau_B}1_{\{\tau_B\leq T\}}\big]}
                  {\sum_{i=1}^N E\big[\e^{-rT_i}1_{\{\tau_B> T_i\}}\big]},
\end{align}
where $r$ denotes the risk-free interest rate.

Now, using that $1_{\{\tau_B\leq t\}}=1_{\{\underline S_t\leq B\}}$ which
immediately translates into
\begin{align}\label{ST}
P(\tau_B\leq t)
 = E\big[1_{\{\tau_B\leq t\}}\big]
 = E\big[1_{\{\underline S_t\leq B\}}\big],
\end{align}
and that
\begin{align*}
E\big[\e^{-r\tau_B}1_{\{\tau_B\leq T\}}\big]
 = \int_0^T\e^{-rt}P_{\tau_B}(\ud t),
\end{align*}
the quantities in \eqref{EDS} can be calculated using the valuation
formulas for one-touch options.

\bibliographystyle{chicago}
\bibliography{references}

\end{document}